\PassOptionsToPackage{expansion=false}{microtype}
\documentclass[sigconf,nonacm]{acmart}

\AtBeginDocument{%
  \providecommand\BibTeX{{\normalfont B\kern-0.5em{\scshape i\kern-0.25em b}\kern-0.8em\TeX}}}
\setcopyright{none}
\acmConference[KDD '27]{Proceedings of the 33rd ACM SIGKDD Conference on
  Knowledge Discovery and Data Mining}{2027}{}
\acmISBN{978-x-xxxx-xxxx-x/27/xx}
\acmDOI{10.1145/nnnnnnn.nnnnnnn}
\usepackage{booktabs}
\usepackage{amsmath}
\usepackage{tikz}
\usetikzlibrary{arrows.meta,positioning,calc,fit}
\usepackage{multirow}

\begin{document}

\title{Beyond Participant-Level Cross-Validation: Reliable Inference
for Longitudinal Machine Learning}

\author{Shahran Rahman Alve}
\affiliation{%
  \institution{Department of Computer Science, The University of Texas at Dallas}
  \city{Richardson}
  \state{Texas}
  \country{USA}}
\email{shahranrahman.alve@utdallas.edu}

\renewcommand{\shortauthors}{SR Alve}

\begin{abstract}
Longitudinal sensing studies routinely collect thousands of windows from a few
dozen participants. The records are numerous; the independent scientific units
are not. When the outcome is defined per participant, this mismatch makes
apparently precise findings vulnerable to pseudo-replication, to partition
choice, and to the ordinary analytic flexibility of comparing several pipelines
before reporting one. Splitting on participants prevents a person's records from
straddling a split, but it does not calibrate the label-dependent workflow fold construction, preprocessing, tuning, calibration, and candidate selection that produced the reported number. We define a participant-level estimand and
obtain an analysis-matched null by permuting participant labels and rerunning
that entire workflow. In controlled simulation, window-level inference rejects in
70-80\% of replicates when no effect exists and a window bootstrap rejects at the
same rate; a participant bootstrap still rejects at 10-17\%; the analysis-matched
test holds 0.025-0.100 across cohorts of 20 to 80 participants. Freezing the
selected pipeline instead of repeating the search inflates Type-I error to 0.240
with eight candidates, where repeating it holds 0.040. Applied to two public
cohorts, wrist actigraphy ($n{=}55$) yields participant AUROC 0.928 with
$p = 0.0050$, a conclusion that persists under a scale-robust rank-pooled statistic and under a
matched permutation null computed after excluding hospitalised participants
($p = 0.0089$).
Smartphone sensing ($n{=}38$, 7 positives) yields 0.636 and does not reject
($p = 0.1724$) despite sufficient resolution, with sensitivity 0.143. The
practical rule is narrow: every step that reads labels belongs inside the permuted
analysis, and repeated records do not create additional independent participants.
\end{abstract}



\maketitle
\pagestyle{empty}
\thispagestyle{empty}

\section{Introduction}
Passive sensing and digital phenotyping aim to infer participant-level health
outcomes from repeated measurements such as actigraphy, sleep, mobility, and
phone interaction. A recent cross-condition scoping review of 65 smartphone
sensing studies reports a median cohort of 52 participants
(IQR 26--126)~\cite{dumas2026scoping}, while individual studies may contain
thousands of windows. This mismatch between the number of records and the number
of independent observational units creates a difficult evaluation problem.
Participant-level splitting prevents one person's windows from appearing on both
sides of a split, but it does not by itself make a high point estimate stable or
informative.

The problem is amplified by ordinary analytical flexibility. A study may compare
several feature sets, model classes, calibration procedures, and partitions
before reporting one number. In a small cohort, selecting the most favourable
pipeline can produce a large AUROC even when labels are unrelated to features.
A chance reference of 0.5 does not describe the distribution induced by the
complete evaluation design, and an interval computed around one fitted pipeline
need not account for the selection that produced it.

Permutation testing addresses this directly. Under participant-level
exchangeability, labels may be reassigned across participants while each
participant's full longitudinal feature history is preserved. The resulting
distribution is interpretable only when every data-dependent step used for the
reported result is repeated under permutation. We therefore treat preprocessing,
tuning, calibration, and model selection as part of the test statistic rather
than as fixed preliminaries.

Our contributions are:
\begin{enumerate}\itemsep1pt \parskip0pt
\item An analysis-matched participant-level permutation test for grouped
      longitudinal machine learning that treats the complete prespecified analysis --- preprocessing,
      tuning, calibration, label-dependent fold construction, and optionally
      candidate selection --- as the test statistic (Section~\ref{sec:method}).
\item Two simulations with known ground truth: one comparing Type-I error and
      power against the window-level and bootstrap procedures a practitioner
      would otherwise use, and one isolating what happens when the selection step
      is left outside the null (Section~\ref{sec:sim}).
\item A scientific application to two public sensing cohorts whose conclusions
      diverge under the same analysis, with a confound-restricted matched null for
      the inpatient-excluded cohort, repeated partition schedules, and a
      scale-robust sensitivity analysis.
\end{enumerate}

\section{Related Work}

\textbf{Splitting, leakage, and identity confounding.} The closest prior work
comes from digital health rather than machine learning venues.
\citet{saeb2017needtoconsider} showed that assigning a subject's repeated records
to both training and test sets can drastically underestimate error, because the
classifier learns to recognise people rather than the condition.
\citet{chaibubneto2019identity} turned that into a test: a permutation procedure
quantifying how much identity confounding a record-wise classifier has absorbed.
Our work begins where theirs ends. They ask whether a \emph{record-wise} design is
contaminated by identity; we assume the subject-wise design they recommend and ask
whether such a result survives a null that also absorbs the analyst's tuning and
selection. The tests are complementary and target different failure modes.

Leakage of many kinds drives irreproducibility across fields
\citep{kapoor2023leakage}, with cases in brain MRI \citep{yagis2021leakage},
digital pathology \citep{bussola2021leakage}, COVID-19 imaging
\citep{roberts2021commonpitfalls}, and medical imaging broadly
\citep{varoquaux2022medicalimaging}; \citet{brown2023shortcut} show models routing
around the intended signal. In each case the defect is invisible in the headline
metric and appears only under a designed check.

\textbf{Small samples and overoptimism.} \citet{varoquaux2018crossvalidation}
quantifies how wide cross-validation error bars become in small samples, and
\citet{nadeau2003inference} explain why naive variance estimates for
cross-validated error are anticonservative when folds share training data.
\citet{berisha2022overoptimistic} find an inverse relationship between sample size
and reported accuracy in clinical speech ML, the signature of a literature
reporting the tail of a noisy distribution;
\citet{goetz2024generalization,jeong2025personalized,mehta2026gmlp} document
related generalisation and practice gaps.

\textbf{Cohort sizes in sensing-based health ML.} A 2026 scoping review of 65
smartphone digital-phenotyping studies reports a median cohort of 52
participants, IQR 26--126 \citep{dumas2026scoping}; a review of passive sensing in
psychosis finds heterogeneous methods and limited external validation
\citep{bladon2025passive}, and recent deployments remain in the same range
\citep{kim2026korean}. Surveys of wearable clinical AI note the need for
validation standards \citep{advmat2025wearablesensors,mahajan2025wearableai}, and
\citet{natcomm2026frailty} demonstrate wearable inference at clinical scale. The
resulting mismatch is stark: tens of independent people, thousands of records, and
metrics usually reported over the records.

\textbf{Permutation inference.} Permutation testing of classifier performance is
long established. \citet{ojala2010permutation} set out the null hypotheses and
test statistics, \citet{phipson2010permutation} show why a Monte Carlo permutation
$p$-value needs the $+1$ correction that Eq.~(2) uses, and
\citet{conroy2012fast} combine model selection with permutation testing for
$\ell_2$-regularized logistic regression. Permutation testing of classifier performance is well established. Our
contribution is an analysis-matched formulation for grouped longitudinal machine
learning: a participant-level estimand, a null in which label-dependent fold
construction is regenerated rather than reused, end-to-end repetition of pipeline
selection, and an empirical characterisation of what follows when these steps are
omitted.

\textbf{Datasets.} StudentLife \citep{wang2014studentlife} instrumented one
48-student class over a ten-week term with pre- and post-term instruments
including the PHQ-9 \citep{kroenke2001phq9}; Depresjon
\citep{garciaceja2018depresjon} released one-minute wrist actigraphy for 23
patients and 32 controls.

\textbf{Privacy.} DP-SGD \citep{abadi2016dpsgd} usually accounts at example level,
which under a record-per-window design does not protect a participant. Membership
inference \citep{shokri2017membership} is the empirical counterpart, and leakage
persists at large $\varepsilon$ \citep{bertoli2026nist}.

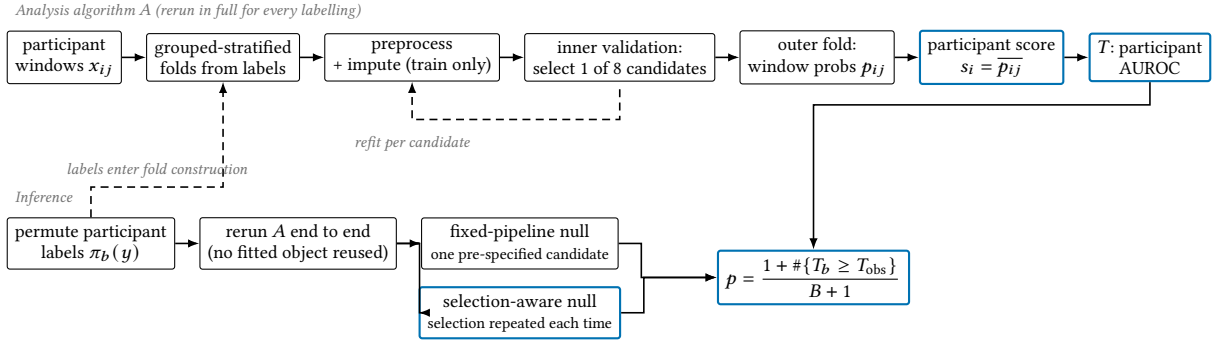
\begin{figure*}[t]
\centering
\begin{tikzpicture}[
  font=\footnotesize,
  node distance=3.2mm,
  box/.style={draw, rounded corners=1pt, minimum height=6mm, align=center,
              inner xsep=3pt, inner ysep=2.5pt},
  acc/.style={box, draw=accent, thick},
  arr/.style={-{Latex[length=1.5mm]}, semithick},
  lab/.style={font=\scriptsize\itshape, text=black!55}
]
\definecolor{accent}{RGB}{0,114,178}

\node[lab, anchor=west] (r1lab) at (0,0) {Analysis algorithm $A$ (rerun in full for every labelling)};
\node[box, below=2.5mm of r1lab.west, anchor=north west] (data)
  {participant\\windows $x_{ij}$};
\node[box, right=of data] (folds) {grouped-stratified\\folds from labels};
\node[box, right=of folds] (prep) {preprocess\\+ impute (train only)};
\node[box, right=of prep] (inner) {inner validation:\\select 1 of 8 candidates};
\node[box, right=of inner] (outer) {outer fold:\\window probs $p_{ij}$};
\node[acc, right=of outer] (score) {participant score\\$s_i=\overline{p_{ij}}$};
\node[acc, right=of score] (stat) {$T$: participant\\AUROC};

\foreach \a/\b in {data/folds, folds/prep, prep/inner, inner/outer, outer/score, score/stat}
  \draw[arr] (\a) -- (\b);

\draw[arr, densely dashed] (inner.south) ++(0,-1mm) -- ++(0,-4mm) -| ($(prep.south)+(0,-5mm)$)
  -- (prep.south) node[lab, midway, below, yshift=-3.5mm] {refit per candidate};

\node[lab, anchor=north west] (r2lab) at ($(data.south west)+(0,-13mm)$)
  {Inference};
\node[box, below=2.5mm of r2lab.west, anchor=north west] (perm)
  {permute participant\\labels $\pi_b(y)$};
\node[box, right=of perm] (rerun) {rerun $A$ end to end\\(no fitted object reused)};
\node[box, right=of rerun, minimum width=26mm] (fixnull)
  {fixed-pipeline null\\\scriptsize one pre-specified candidate};
\node[acc, below=2.2mm of fixnull, minimum width=26mm] (selnull)
  {selection-aware null\\\scriptsize selection repeated each time};
\node[acc, right=13mm of fixnull, yshift=-4.5mm] (pval)
  {$p=\dfrac{1+\#\{T_b\ge T_{\mathrm{obs}}\}}{B+1}$};

\draw[arr] (perm) -- (rerun);
\draw[arr] (rerun.east) -- ++(3mm,0) |- (fixnull.west);
\draw[arr] (rerun.east) -- ++(3mm,0) |- (selnull.west);
\draw[arr] (fixnull.east) -- ++(3mm,0) |- (pval.west);
\draw[arr] (selnull.east) -- ++(3mm,0) |- (pval.west);
\draw[arr] (stat.south) -- ++(0,-3mm) -| ($(pval.north)+(0,2mm)$) -- (pval.north);
\draw[arr, densely dashed] (perm.north) -- ++(0,4mm) -| (folds.south)
  node[lab, pos=0.25, above] {labels enter fold construction};
\end{tikzpicture}
\caption{The analysis algorithm $A$ (top) maps participant windows to a single
participant-level statistic, with every data-dependent step---label-dependent fold
construction, preprocessing, and candidate selection by inner validation---fitted
inside training participants. Inference (bottom) permutes participant labels and
reruns $A$ from the beginning. The two branches differ only in what $A$ contains:
the \emph{fixed-pipeline} null evaluates one pre-specified candidate, whereas the
\emph{selection-aware} null repeats the candidate search under every permutation.
The gap between them is the optimism that a fixed-pipeline null does not capture.
Because folds are grouped-stratified, labels enter the design itself (dashed
arrow), so folds are regenerated from each permuted labelling.}
\label{fig:arch}
\end{figure*}

\section{Method}\label{sec:method}

\subsection{Participant-level estimand}
Dataset $d$ contains participants $i = 1,\dots,n_d$ with outcome $y_i$ and
eligible windows $x_{ij}$, $j = 1,\dots,m_i$; all splits are over participants.
Within an outer training set, imputation, scaling, feature handling,
hyperparameter selection, calibration, and threshold selection are fitted without
access to held-out participants, and the fitted pipeline produces window
probabilities $p_{ij}$. The primary participant score is
\begin{equation}
s_i \;=\; m_i^{-1} \textstyle\sum_{j=1}^{m_i} p_{ij},
\end{equation}
with mean aggregation fixed before evaluation. All primary metrics are computed
from the $n_d$ pairs $(y_i, s_i)$, so no participant is weighted more heavily for
contributing more windows. Median and most-recent-window aggregation are
sensitivity analyses only.

\subsection{Selection-aware permutation test}
Let $A$ denote the complete analysis mapping a labelled participant dataset to a
pre-specified statistic $T$, including all data-dependent preprocessing, tuning,
calibration, and selection among candidate pipelines. We compute
$T_{\text{obs}} = A(X, y)$. For each of $B$ participant-level permutations
$\pi_b$, all windows of participant $i$ receive the common reassigned label
$y_{\pi_b(i)}$, and we recompute $T_b = A(X, \pi_b(y))$ from the beginning. No
fitted preprocessor, hyperparameter, calibrator, or selected model is carried over
from the observed labelling. The one-sided Monte Carlo $p$-value is
\begin{equation}
p \;=\; \frac{1 + \sum_{b=1}^{B} \mathbf{1}\!\left[T_b \ge T_{\text{obs}}\right]}{B + 1},
\end{equation}
reported with a binomial confidence interval for finite-permutation uncertainty.
Permutation quantiles are used only for visualisation; inference uses the
$p$-value. Across primary datasets, $p$-values are adjusted by Holm's method.

\subsection{Validity}
\begin{proposition}
Let $A$ be a fully specified, possibly randomized analysis mapping a labelled
participant dataset to a scalar $T$. Under the null that participant labels are
exchangeable conditional on the participant feature histories, the permutation
test that recomputes $A$ under every participant-label permutation controls the
Type-I error at level $\alpha$, provided $A$ and its randomization policy are
applied identically to the observed and to every permuted dataset.
\end{proposition}

\noindent\emph{Proof sketch.} Conditional on the feature histories and the seed
schedule, exchangeability makes the observed labelling and its permutations
equally likely; applying the same $A$ to each leaves the multiset of statistics
permutation invariant, so the rank of $T_{\text{obs}}$ is uniform up to ties and
the $+1$-corrected $p$-value is conservative under Monte Carlo sampling.
Data-dependent preprocessing, tuning, calibration, and selection do not invalidate
the test because they sit inside $A$.\hfill$\square$

\smallskip
\noindent \textbf{Exchangeability in these cohorts.} Depresjon was released as a
clinical case series with unmatched controls and no site, family, or pairing
structure documented in the release, so we use unrestricted participant-level
permutation under the exchangeability assumption for that cohort structure. One
observed design feature that may challenge exchangeability is hospitalisation
status, which we therefore examine in a separately permuted restricted analysis
(\S\ref{sec:confound}). StudentLife participants were observed over the same
academic term, so calendar time is shared at the cohort level rather than defining
known participant-specific permutation blocks. Designs with known blocks --- sites,
families, matched pairs, or treatment strata --- require restricted permutation
within those blocks, which is why the recording-length-matched analysis would
require permutation within its matched strata.

\noindent The assumptions are substantive. Labels must be exchangeable across
participants; site, family, or temporal block structure would require permutation
within blocks. $A$ must be frozen before the final test, and the procedure
calibrates that frozen $A$: it cannot retrospectively correct for choices made
while developing $A$ after seeing earlier results, for which only pre-registration
or a held-out cohort suffices. A chance-centred null does not certify that
features are free of label proxies, and the test concerns association under the
design, not causality or clinical validity.

\begin{figure*}[t]
\centering
\includegraphics[width=0.78\textwidth]{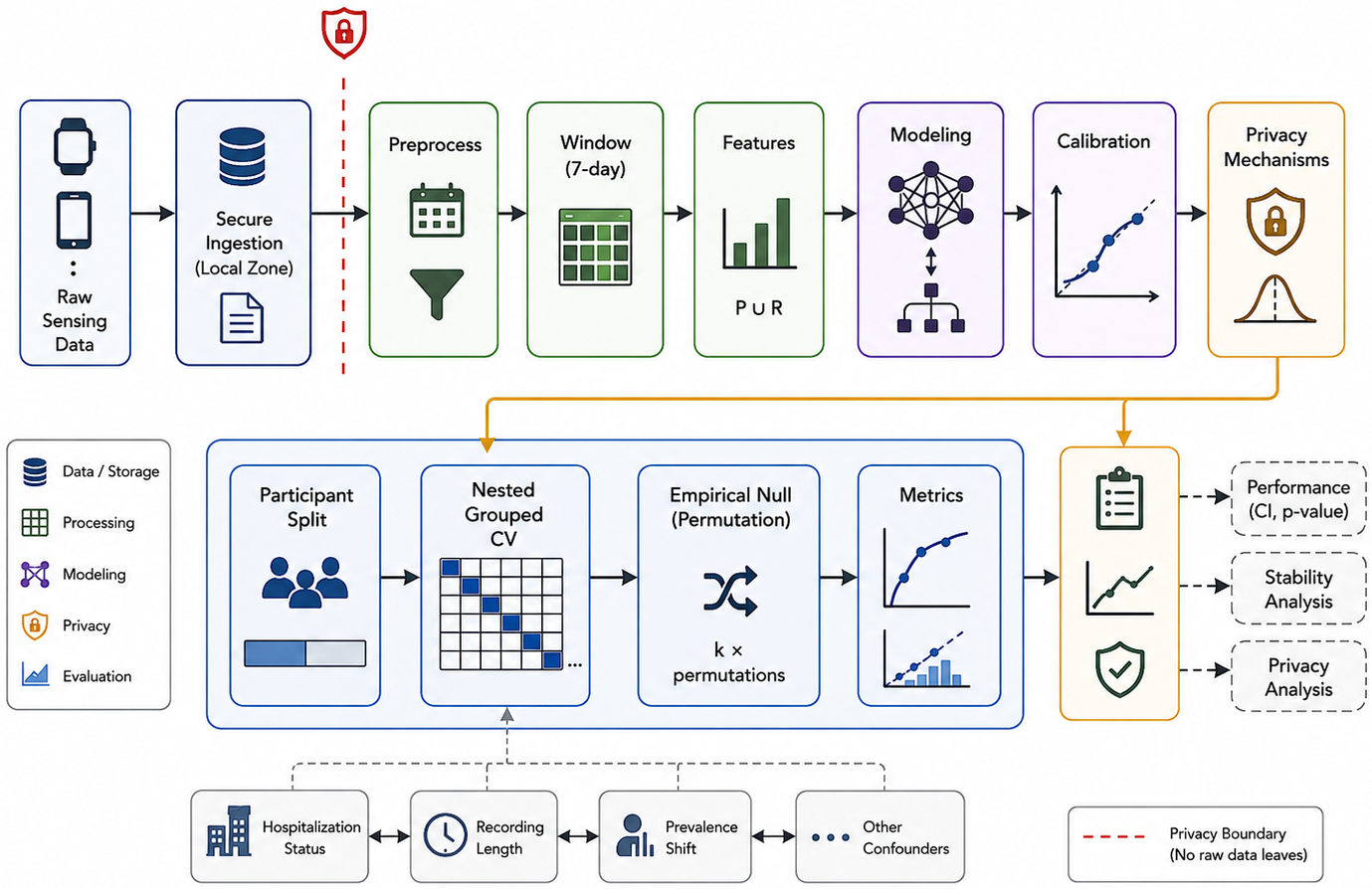}
\caption{The analysis sits inside a larger system: raw records, coordinates, and
exact timestamps do not cross the privacy boundary (dashed). Confound analyses
re-enter the cross-validation loop, each requiring its own matched null.
Figure~\ref{fig:arch} remains authoritative for the inference procedure.}
\label{fig:overview}
\end{figure*}

\section{Experimental Setup}\label{sec:setup}
Depresjon yields 55 participants (23 positive) and 788 windows; StudentLife under
a post-term PHQ-9 cutoff of 10~\cite{kroenke2001phq9} yields 38 participants with
a usable label, of whom 7 are positive, and 1{,}938 windows. Windows are seven-day
histories built only from days at or before each cutoff.

Outer evaluation uses five grouped folds; every participant is held out exactly
once, with complete out-of-fold coverage in both datasets. Two statistics are
reported, differing only in what $A$ contains. The
\emph{selection-aware} statistic is the participant-level pooled
out-of-fold AUROC obtained when, within each outer training fold, the candidate is
chosen by inner-validation participant AUROC among eight pre-specified pipelines:
$\ell_2$ logistic regression at two regularisation strengths on the combined
feature set and at one strength on each of the population-only and
participant-relative branches, two random forests differing in minimum leaf size,
and two histogram gradient-boosting configurations differing in learning rate and
leaf count. Ties break by candidate name. The \emph{fixed-pipeline} statistic uses only the
pre-specified gradient-boosting configuration and is an ablation showing what
ignoring selection misses. The candidate set was frozen before the final analysis
and is identical under every permutation and partition seed.

For each permuted labelling the grouped-stratified outer and inner folds are
regenerated from the permuted labels using the same deterministic seed schedule
and splitting algorithm; no fold assignment constructed from the observed labels
is reused, because stratification reads the labels and is therefore part of $A$.

Confidence intervals are percentile intervals from a stratified participant
bootstrap with 10{,}000 replicates, resampling positives and negatives separately
so every replicate contains both classes. Thresholds are chosen on inner
validation participants only; with seven positive StudentLife participants, all
thresholded metrics there are highly uncertain.

\section{Simulation Study}\label{sec:sim}
Everything else here measures a real cohort where the truth is unknown. To check
that the procedure does what it claims, we need data whose answer we already know.

We simulate participants contributing 8--30 correlated windows over 12 features.
A participant random intercept (SD 1.0) is deliberately large relative to window
noise (SD 0.6), reproducing the property that causes trouble in practice: two
windows from one person resemble each other far more than either resembles anyone
else's. A latent effect $\delta$ shifts four of the twelve features for positive
participants, so $\delta = 0$ is the null. Prevalence is 42\%, roughly matching the
datasets in Section~\ref{sec:setup}.

Four procedures see identical data. \textbf{Window-naive} tests window-level AUROC
against 0.5 with a Mann--Whitney test treating every window as independent.
\textbf{Window bootstrap} resamples windows and rejects when the 95\% interval
excludes 0.5. \textbf{Participant bootstrap} aggregates to one score per
participant and resamples participants. \textbf{Permutation} is the test of
Section~\ref{sec:method} at 99 permutations per replicate. All four share the same
nested grouped cross-validation and regularized logistic classifier; only the
inference differs. A linear learner keeps $10^{4}$ fits affordable, and the
question is whether the \emph{test} controls error, not which classifier wins.

\begin{table*}[t]
\centering\small
\caption{Rejection rates at a nominal 5\% level, 40--60 simulated replicates per
cell. At $\delta = 0$ the rate is Type-I error and should sit near 0.05; above zero
it is power. All four procedures see identical data, identical grouped folds, and
the same classifier; only the inference differs. Brackets are 95\%
Clopper--Pearson intervals.}
\label{tab:sim}
\begin{tabular}{lccccccccc}
\toprule
& \multicolumn{3}{c}{$n = 20$ participants} & \multicolumn{3}{c}{$n = 40$}
& \multicolumn{3}{c}{$n = 80$} \\
\cmidrule(lr){2-4}\cmidrule(lr){5-7}\cmidrule(lr){8-10}
Procedure & $\delta{=}0$ & $\delta{=}0.3$ & $\delta{=}0.6$
          & $\delta{=}0$ & $\delta{=}0.3$ & $\delta{=}0.6$
          & $\delta{=}0$ & $\delta{=}0.3$ & $\delta{=}0.6$ \\
\midrule
Window-naive          & 0.700 & 0.825 & 0.775 & 0.767 & 0.833 & 0.967 & 0.800 & 0.775 & 0.975 \\
Window bootstrap      & 0.675 & 0.825 & 0.775 & 0.767 & 0.817 & 0.967 & 0.800 & 0.775 & 0.975 \\
Participant bootstrap & 0.150 & 0.125 & 0.375 & 0.167 & 0.233 & 0.517 & 0.100 & 0.150 & 0.950 \\
\textbf{Permutation}  & \textbf{0.100} & 0.100 & 0.250
                      & \textbf{0.033} & 0.133 & 0.517
                      & \textbf{0.025} & 0.175 & 0.925 \\
 & \tiny[.03,.24] & & & \tiny[.00,.12] & & & \tiny[.00,.13] & & \\
\midrule
Mean participant AUROC & 0.534 & 0.557 & 0.646 & 0.483 & 0.550 & 0.685 & 0.477 & 0.554 & 0.736 \\
\bottomrule
\end{tabular}
\end{table*}

Table~\ref{tab:sim} is the central result of this paper. Under the null, the
window-naive procedure rejects in 70--80\% of replicates at a nominal 5\% level,
and the rate does not improve with cohort size, because adding participants does
not fix an analysis that counts windows. Although window bootstrapping accounts for
sampling variability, it rejects at essentially the same rate: resampling windows
does not restore independence when the participant is the observational unit. Aggregating to participants helps substantially but does not suffice;
the participant bootstrap still rejects at 10--17\%, two to three times nominal,
because a bootstrap interval around a cross-validated statistic does not account
for the fitting that produced it. The permutation test rejects at 0.025--0.100
across the three cohort sizes, with intervals covering the nominal level
throughout.

The natural objection to a conservative test is that it trades power for error
control. That trade is real but small, and it shrinks as the cohort grows. At
$n = 40$ and $\delta = 0.6$ the permutation test and the participant bootstrap
reject at exactly the same rate (0.517), while only one of them controls Type-I
error. At $n = 80$ the gap is 0.925 against 0.950. At $n = 20$ the permutation test is
meaningfully less powerful (0.250 against 0.375), reflecting the power cost of
maintaining Type-I error control in the smallest cohorts; the participant
bootstrap's higher apparent power there should be interpreted alongside its
elevated Type-I error of 0.150 under the null.

Two further points. Mean participant AUROC under the null sits at 0.477--0.534, so
the point estimate gives no warning; the inflation lives entirely in the
inference. And at $n = 20$ every procedure struggles, ours included, whose Type-I
rate of 0.100 carries an interval of $[0.03, 0.24]$. Cell sizes of 40--60
replicates make these intervals wide enough that we describe the ordering rather
than exact values, and the simulation uses a linear learner and Gaussian
covariance, so it establishes error control under a plausible grouped design
rather than every design.

\subsection{Does the selection step have to be inside the null?}
Rerunning candidate selection under every permutation costs more than permuting
after a model has been chosen; Table~\ref{tab:selnull} asks whether that buys
anything. Under the null we vary the number of candidates $K$ and compare two
reference distributions for the same observed statistic: the \emph{fixed-pipeline} variant, which refits only the winning pipeline, and
the \emph{selection-aware} variant, which reruns the search each time.

\begin{table}[t]
\centering\small
\caption{Selection inside the null, at $\delta = 0$ (no effect present), 40
participants, 25--30 replicates per row. The \emph{fixed-pipeline} variant permutes labels after the
winning pipeline has been chosen; the \emph{selection-aware} variant reruns the
candidate search inside every permutation. At $K = 1$ there is nothing to select and the two
coincide, as they must. As the candidate set grows the fixed-pipeline
variant over-rejects while the selection-aware variant stays near nominal.}
\label{tab:selnull}
\begin{tabular}{lccc}
\toprule
Candidates $K$ & 1 & 4 & 8 \\
\midrule
Type-I, fixed-pipeline & 0.067 & 0.100 & \textbf{0.240} \\
Type-I, selection-aware& 0.067 & 0.033 & \textbf{0.040} \\
\midrule
Null mean, fixed       & 0.487 & 0.484 & 0.479 \\
Null mean, sel.-aware  & 0.487 & 0.556 & 0.557 \\
Null 95th, fixed       & 0.677 & 0.667 & 0.667 \\
Null 95th, sel.-aware  & 0.677 & 0.715 & 0.719 \\
Upper-tail shift       & 0.000 & 0.048 & 0.052 \\
\bottomrule
\end{tabular}
\end{table}

At $K = 1$ there is nothing to select and the two nulls coincide exactly, the
sanity check the design requires. From there they separate: with four candidates
the fixed-pipeline variant rejects at 0.100 and with eight at 0.240, while the
selection-aware variant holds 0.033 and 0.040. The lower rows show the mechanism --- repeating the search
lifts the null's 95th percentile by roughly 0.05, so a study comparing against the
fixed-pipeline distribution uses a reference that is too low by about that much. Permuting
labels after choosing a model tests whether that pipeline beats chance, which is
not the question asked when the reported number came from a search.

\section{Results}\label{sec:results}

\subsection{Participant-level performance}
Table~\ref{tab:main} reports the primary results. The participant-level estimand
changes what is reported: on Depresjon the window-pooled estimate was 0.830 against
0.928 at participant level, while on StudentLife an AUPRC of 0.376 against a 0.184
prevalence baseline, with sensitivity 0.143, shows what a moderate AUROC conceals
at 18\% prevalence.

\begin{table}[t]
\centering\small
\caption{Participant-level out-of-fold performance. Intervals are stratified
participant bootstrap percentile intervals, resampling positives and negatives
separately so every replicate contains both classes. Thresholded metrics use a
threshold selected on inner validation participants only; the intervals do not
include threshold-selection variability, and the StudentLife threshold rests on
seven positive participants.}
\label{tab:main}
\begin{tabular}{lcc}
\toprule
 & Depresjon & StudentLife \\
\midrule
Participants (pos / neg)      & 55 (23/32) & 38 (7/31) \\
AUROC                         & 0.928 & 0.636 \\
\quad 95\% CI                 & [0.851, 0.980] & [0.388, 0.858] \\
AUPRC (prevalence)            & 0.918 (0.418) & 0.376 (0.184) \\
\quad 95\% CI                 & [0.822, 0.978] & [0.122, 0.692] \\
Brier score                   & 0.115 & 0.163 \\
Log loss                      & 0.385 & 0.815 \\
Balanced accuracy             & 0.791 & 0.539 \\
\quad 95\% CI                 & [0.676, 0.894] & [0.436, 0.698] \\
Sensitivity                   & 0.739 & 0.143 \\
\quad 95\% CI                 & [0.565, 0.913] & [0.000, 0.429] \\
Specificity                   & 0.844 & 0.936 \\
\quad 95\% CI                 & [0.719, 0.969] & [0.839, 1.000] \\
PPV                           & 0.773 & 0.333 \\
\quad 95\% CI                 & [0.625, 0.933] & [0.000, 1.000] \\
Positives detected            & 17 of 23 & \textbf{1 of 7} \\
\bottomrule
\end{tabular}
\end{table}

\subsection{Permutation tests}
Table~\ref{tab:null} reports both nulls. On Depresjon no permuted statistic
reached the observed value across $B = 201$ fixed-pipeline permutations, giving
$p = 0.0050$ with a Monte Carlo interval of $[0.005, 0.023]$. On StudentLife 9 of
57 permuted statistics matched or exceeded the observed value, giving $p = 0.1724$
$([0.086, 0.294])$; the exchangeability null is not rejected. Because the
StudentLife smallest attainable $p$-value is 0.017, that test had the resolution
to reject and did not, which makes this a non-rejection rather than an
inconclusive test.

The Depresjon null mean of 0.4993 is what an exchangeability null should produce.
An earlier version reported 0.471, and the discrepancy was diagnostic.
Grouped-\emph{stratified} fold construction reads the labels and therefore belongs
inside $A$; holding folds fixed while permuting leaves every permuted fit
stratified against a labelling that no longer exists. Regenerating folds from each
permuted vector moved the null mean from 0.471 to 0.499 and the maximum from 0.821
to 0.709. All results use the regenerated design. Centring provides no evidence of label leakage detectable by this diagnostic; it
does not certify its absence.

\emph{Exploratory.} The lower block of Table~\ref{tab:null} compares the fixed-pipeline and
selection-aware statistics on the \emph{same} permutation identifiers. Repeating
selection raises the null mean by 0.044 (95\% bootstrap CI $[0.022, 0.068]$),
positive in 78\% of permutations (sign test $p=0.011$), and moves the 95th
percentile from 0.676 to 0.715. At matched $B=23$, however, both tests return
$p=0.042$: the shift moved the reference distribution without changing the
inference, because the observed statistic lies far above both nulls. Comparing $p$-values obtained at different $B$ would confound a selection effect
with Monte Carlo resolution, so the contrast above is computed on matched
permutations only. Nor is a null-mean shift itself evidence of optimism:
inner selection within outer training folds is proper nested evaluation. The
quantity that would change inference is the upper tail, estimated too imprecisely
at $B=23$, and we have not excluded that outer folds selecting different model
families introduce a score-scale artefact when pooled. We report the shift as
requiring confirmation at $B\ge999$.

\begin{table}[t]
\centering\small
\caption{Participant-label permutation tests. Upper block: fixed-pipeline test
per dataset. Lower block: paired fixed-versus-selection-aware comparison on
Depresjon computed on the \emph{same} permutation identifiers, so the contrast is
free of the Monte Carlo-resolution artefact that comparing $p$-values at different
$B$ would introduce. Intervals for $p$ are Clopper--Pearson for the unknown tail
probability.}
\label{tab:null}
\begin{tabular}{lcc}
\toprule
 & Depresjon & StudentLife \\
\midrule
$T_{\text{obs}}$              & 0.928 & 0.636 \\
$B$                            & 201 & 57 \\
Null mean (SD)                 & 0.490 (0.105) & 0.482 (0.150) \\
Null max                       & 0.769 & 0.908 \\
$\#\{T_b \ge T_{\text{obs}}\}$ & 0 & 9 \\
Monte Carlo $p$                & \textbf{0.0050} & 0.1724 \\
\quad 95\% MC interval         & [0.005, 0.023] & [0.086, 0.294] \\
Smallest attainable $p$        & 0.0050 & 0.017 \\
Rank-pooled $p$ (\S\ref{sec:robust})  & 0.0090 & --- \\
\midrule
\multicolumn{3}{l}{\emph{Paired, Depresjon, matched $B=23$}} \\
 & Fixed & Selection-aware \\
Null mean                      & 0.502 & 0.564 \\
Null 95th / 99th pct           & 0.676 / 0.719 & 0.715 / 0.733 \\
Exceedances                    & 0 & 0 \\
Matched $p$                    & 0.042 & 0.042 \\
$\Delta_b$ mean [95\% CI]      & \multicolumn{2}{c}{0.044 [0.022, 0.068]} \\
$\Pr(\Delta_b>0)$              & \multicolumn{2}{c}{0.78 ($p=0.011$)} \\
\bottomrule
\end{tabular}
\end{table}

\subsection{Partition stability}
We regenerated the grouped-stratified five-fold design under independent seeds and
reran the complete procedure for each (Figure~\ref{fig:stability}). On Depresjon
the fixed pipeline has median 0.909 across 23 seeds, 5th--95th percentile
$[0.857, 0.936]$, range 0.095: across the evaluated seeds the estimate varied
little. On StudentLife it has median 0.712 across 6 seeds, $[0.624, 0.803]$, range 0.194
--- roughly twice as wide across the seeds evaluated, spanning values a study would
present very differently. Larger sweeps are needed to estimate partition-induced
variability precisely. Across the six evaluated StudentLife seeds the selection-aware median was about
0.04 above the fixed-pipeline median (0.753 against 0.712), while on Depresjon the
two medians coincide at 0.909--0.910; three different candidates win across the
six StudentLife seeds. This comparison is exploratory and too small to
characterise the selection effect precisely. Seed counts are below the
50--100 we would report, and these are differences of medians rather than paired
seed-level effects.

\begin{figure}[t]
\centering
\includegraphics[width=\columnwidth]{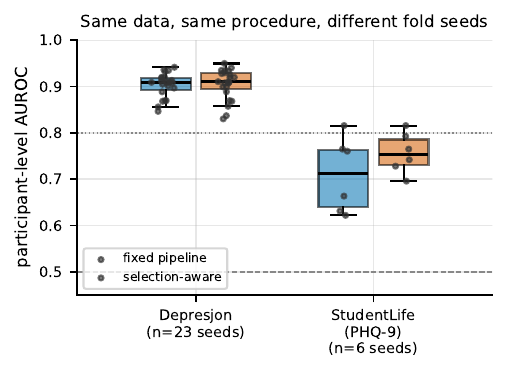}
\caption{Participant-level AUROC across independent grouped fold seeds with the
procedure otherwise unchanged. The dotted line marks 0.80. Depresjon is stable;
StudentLife spans roughly 0.62 to 0.82, so a single reported partition
underdetermines the conclusion. Seed counts differ because each StudentLife seed
is substantially more expensive.}
\label{fig:stability}
\end{figure}

\subsection{Is the pooled statistic an artefact of cross-fold scales?}\label{sec:robust}
Pooling out-of-fold probabilities assumes they are comparable across folds, which
is not automatic when different folds select different model families. We
therefore replace each participant's score with its quantile rank \emph{within its
own outer fold}, making the statistic invariant to any strictly monotone per-fold
rescaling, and recompute the permutation null under that definition rather than
re-scoring the existing one. On Depresjon the rank-pooled statistic is 0.928,
matching the pooled value to three decimals, with $p = 0.0090$ ($B = 110$, no
exceedances, null mean 0.511): the primary conclusion does not depend on
cross-fold score comparability.

\subsection{Do confounders explain the actigraphy result?}\label{sec:confound}
Two features of the Depresjon cohort could produce an association without any
behavioural signal. Five of 23 condition participants were inpatients, so ward
routine rather than depression could drive their activity; and condition
participants recorded fewer days than controls (17.1 versus 22.7), so anything
tracking recording length could carry label information.

A null computed on the full cohort does not transfer to a restricted one, so each
restriction is rerun end to end with its own permutation null under the identical
restricted analysis. Excluding hospitalised participants leaves 50 participants
(18 positive) and gives AUROC 0.927 with $p = 0.0089$ ($B = 112$, no exceedances,
null mean 0.498). The association persists after excluding hospitalised
participants, although this sensitivity analysis does not eliminate other measured
or unmeasured confounding. Matching on recording length gives 0.890 as a point
estimate; because matching creates strata, its null would have to be permuted
within matched blocks, and that block-restricted null is not complete, so we
report the estimate descriptively rather than reuse an unmatched null. These
analyses address the two confounders specified in advance.

\subsection{Privacy case study}
An illustrative DP-SGD \emph{training} study under participant add/remove
adjacency (per-participant averaged gradients clipped to $C{=}1$ then noised;
batches of 8, sampling rate 0.25, 240 steps, SGD at lr 0.5, Opacus 1.6.0 R\'enyi
accountant, $\delta = 2.33{\times}10^{-3}$) separates the mechanism: clipping alone
cost 0.009 AUROC, the noise arms 0.12--0.30. Three seeds cannot support a
privacy-utility curve, and the budget covers training only, since threshold
selection and calibration read the same private data. Full results are in the
appendix.

\section{Discussion}
Longitudinal sensing studies can be data-rich and participant-poor at the same
time, and inference has to respect the participant as the independent unit. Small
cohorts are not unusable; their estimates must be calibrated against the analysis
that produced them. The simulation makes the cost of getting this wrong
concrete: a window-level analysis of data containing no effect at all rejects
three times out of four. This mechanism could contribute to the patterns of
overoptimistic performance reported in small-sample clinical machine learning
\citep{berisha2022overoptimistic}. Participant-level splitting removes overlap, and the
permutation distribution quantifies how much apparent discrimination arises from
labelling, partitioning, tuning, calibration, and selection acting together. The
two cohorts separate: the actigraphy estimate varied little across the evaluated
fold schedules and rejects the exchangeability null, whereas the smartphone
estimate varies roughly twice as widely and does not. AUROC 0.636 with sensitivity
0.143, recovering one of seven positives, shows why a single discrimination number
is insufficient at low prevalence.

Our test complements the identity-confounding test of
\citet{chaibubneto2019identity} rather than replacing it: theirs asks whether a
record-wise design absorbed subject identity, ours whether a subject-wise result
survives a null containing the analyst's choices.

Limitations bound every claim above. Permutation counts are below the 999 that would be standard for a headline
inference ($B{=}201$ for Depresjon, $B{=}57$ for StudentLife, $B{=}23$ for the
paired selection contrast). For both datasets the permutation counts are large enough that a 5\% rejection is
attainable in principle (smallest attainable $p$ of 0.0050 and 0.017); only
Depresjon actually rejects. The conclusions are therefore not artefacts of Monte
Carlo resolution, though every interval would tighten with more permutations. A matched permutation null is complete for the inpatient-excluded analysis; the
recording-length-matched analysis remains descriptive because its block-restricted
null is not yet complete, and a null computed on the full cohort does not transfer
to a restricted one. Partition stability rests on 23 and 6 seeds, which supports a qualitative
comparison rather than a distributional claim.
Both datasets are single-site and observational, neither supports powered subgroup
analysis, and the archives were obtained through the project rather than from the
official hosts, so provenance is recorded by hash but not independently verified.

\section{Conclusion}
The problem is not exotic. Record a few dozen people for weeks, cut the recordings
into windows, and the dataset's apparent size exceeds its information content by an
order of magnitude. Splitting on participants handles the crudest version. What it
does not handle is that the analyst then tunes, calibrates, constructs
label-dependent folds, selects, and reports. Each of these label-dependent steps
contributes to the distribution of the reported statistic under the null, yet that
contribution is often absent from standard evaluation.

Our answer is to make the whole procedure the test statistic. In simulation the
resulting test holds Type-I error at 0.025--0.100 where a window-level analysis
rejects 70--80\% of the time under the null, and at $n = 80$ it matches a
participant bootstrap's power to within 0.025. A second simulation shows what
happens when one step is left outside: freezing the selected pipeline rather than
repeating the search inflates Type-I error to 0.240 with eight candidates. On real
data the actigraphy cohort rejects the exchangeability null ($p = 0.0050$),
unchanged under a scale-robust rank-pooled statistic, while the smartphone cohort
does not ($p = 0.1724$) despite having the resolution to.

Two points matter for adoption. For an analysis-matched permutation test, any
label-dependent step used to produce the observed statistic must be reproduced
under permutation: our folds were grouped-stratified,
and holding them fixed biased the null mean to 0.471 instead of 0.490 --- inflating
significance in precisely the direction this paper argues against, and invisible
until we asked why the null was not centred. And a shift in the null mean is not
by itself evidence of optimism, since inner selection within outer training folds
is proper nested evaluation; the quantity that governs a decision is the upper
tail, which is why Table~\ref{tab:selnull} reports Type-I error rather than means
alone. What we can claim is narrow: under this prespecified participant-level
design, one cohort rejects a null matched to the analysis and one does not.

\section{Limitations and Ethical Considerations}
\textbf{Statistical limits.} The permutation test calibrates a \emph{prespecified}
analysis against participant-label exchangeability. It does not repair choices
made while developing that analysis after seeing earlier results; only
pre-registration or a held-out cohort addresses that. A rejection supports an
association under the stated design and does not establish causality,
transportability, fairness, or clinical utility. A non-rejection is not evidence
of no association: it means the observed statistic is not unusually large relative
to the matched null at the available resolution. The test cannot eliminate
unmeasured confounding, and the restricted analyses address only the two confounders specified in
advance. A matched permutation null is complete for the inpatient-excluded
analysis; the recording-length-matched analysis is reported descriptively because
its block-restricted null is not yet complete.

\textbf{Resolution and scope.} Permutation counts are $B = 201$ (Depresjon
fixed-pipeline), 112 (inpatient-excluded), 110 (rank-pooled), 57 (StudentLife),
and 23 for the paired selection-aware contrast, which we report as exploratory
throughout. Simulation cells use 25--60 replicates, so we describe the ordering of
procedures rather than exact rejection rates. Partition stability rests on 23 and
6 seeds respectively, supporting a qualitative comparison rather than a
distributional claim. Both cohorts are single-site and observational, and neither
supports adequately powered subgroup analysis.

\textbf{Data provenance and consent.} Both datasets are public research releases
collected under their originating institutions' approvals; we conducted no new
human-subject data collection and made no attempt at re-identification.
Participant identifiers are salted keyed hashes and the salt is not distributed,
so released artifacts cannot be linked to source participant numbering. The
archives analysed here were obtained through the project rather than downloaded
from the official hosts during this build; file hashes are recorded, but
provenance against the official releases is not independently verified, and users
should confirm the licence terms before redistribution. Derived feature tables are
not redistributed.

\textbf{Potential misuse.} Participant-level risk estimates from behavioural
sensing could be repurposed for surveillance or for consequential decisions in
employment, education, or insurance. We restrict our claims to research
evaluation and recommend no deployment for diagnosis or triage. A more
conservative evaluation standard is itself an ethical contribution: presenting an
unstable model as clinically useful, on the basis of a few dozen participants,
carries real cost for the people such a model would be applied to.

\section{Generative AI Usage}
Generative AI assistance (Anthropic Claude) was used as coding assistant in preparing this work; every reported number is produced by
that code from the recorded artifacts rather than transcribed by hand, and the
Code, configuration files, seed schedules, per-permutation outputs, and dataset acquisition instructions will be made available upon reasonable request.

The author take full responsibility for the content. All statistical claims,
experimental designs, and interpretations were specified and verified by the
authors, and every citation was checked against its source. No text, citation,
result, or dataset detail was accepted without verification. 

{\small
\setlength{\bibsep}{0pt}
\bibliographystyle{ACM-Reference-Format}
\bibliography{references}
}

\end{document}